\documentclass[11pt,fleqn]{article}
\usepackage{amsfonts}
\usepackage{xcolor}
\usepackage[cp1250]{inputenc}
\usepackage{latexsym}
\usepackage{float}
\usepackage{graphicx}
\usepackage{indentfirst}
\usepackage{tabularx}
\usepackage{textcomp}
\usepackage{color}
\usepackage{eucal}
\usepackage{verbatim}
\usepackage[unicode]{hyperref}
\usepackage[colorinlistoftodos]{todonotes}
\hypersetup{colorlinks=false}


\usepackage{amsmath}
\usepackage{amsthm}
\usepackage{amsfonts}
\theoremstyle{definition}
\newtheorem{df}{Definition}[section]
\theoremstyle{plain}
\newtheorem{tw}{Theorem}[section]
\newtheorem{st}{Proposition}[section]

\begin{document}
\numberwithin{equation}{section}
\title{\bf A geometric description of Maxwell field\\ in a Kerr spacetime}
\author{Jacek Jezierski\thanks{E--mail: \texttt{Jacek.Jezierski@fuw.edu.pl}}\; and Tomasz Smołka\thanks{E--mail: \texttt{tksmolka@gmail.com}} \\
	Department of Mathematical Methods in
	Physics, \\ University of Warsaw,
	ul. Pasteura 5, 02-093 Warszawa, Poland}
\maketitle
\begin{abstract}
We consider the Maxwell field in the exterior of a Kerr black hole. For this system, we propose a geometric construction of generalized Klein--Gordon equation called Fackerell--Ipser equation. Our model is based on conformal Yano--Killing tensor (CYK tensor). We present non-standard properties of CYK  tensors in the Kerr spacetime which are useful in electrodynamics.
\end{abstract}
\section{Introduction}
Maxwell's equations in the Kerr spacetime are important for research in relativistic astrophysics. The electrodynamics on Kerr background was examined in 1970s using Newman--Penrose formalism (see 
\cite{Chandra},\cite{Chandra.book}). Fackerell and Ipser reduced Maxwell's equations to a single second order partial differential equation for a complex scalar (see (\ref{FQ_Weyl})). We call it F--I equation. The solution of F-I equation is closely related to the Newman--Penrose electromagnetic scalar\footnote{See equations (\ref{NP}) and (\ref{FI_NP}) for further results.}. Nowadays, F--I equation is investigated in the context of the uniform energy bound and asymptotic behaviour of solutions \cite{And_spinor}.

In this paper we propose a geometric construction which enables one to obtain  F--I equation using conformal Yano--Killing tensor (CYK tensor). CYK tensors were often investigated as a tool to study symmetries and construct conserved quantities (\cite{penrose,Kress,jez-luk}). In electrodynamics there are two kinds of  conserved quantities which are defined with the help of CYK tensors. 
The first one corresponds to electric or magnetic charge and is linear with respect to the Maxwell field. The second kind is quadratic and expresses energy, linear momentum or angular momentum of the electromagnetic tensor.
In our approach we investigate non-standard properties of CYK tensor in Kerr spacetime which 
facilitate reduction of Maxwell's equations to a single generalized Klein--Gordon equation for a complex scalar.

The paper is organized as follows: In Section 2 we review the general properties of CYK tensor and  present a construction of a second order wave equation (\ref{FQ_Riemann})  for special spacetimes equipped with CYK tensor.
The next section is focused on properties of Kerr spacetime. In particular, we demonstrate how equation (\ref{FQ_Riemann}) can be reduced to 
F--I equation (\ref{FQ_Weyl}). During our derivation of F--I equation, additional properties of CYK tensors in Kerr spacetime are presented. To clarify the exposition, some of the technical results and proofs have been shifted to the appendix. The last section contains miscellaneous results: generalization of Fackerell--Ipser equation for Kerr--de Sitter spacetime and discussion of distorted Coulomb solution arising from
electromagnetic field in the Kerr--Newman--de Sitter spacetime.

\section{General properties of Maxwell field and CYK tensor}
Let $M$ be a four-dimensional manifold equipped with pseudo-Riemannian metric $g_{\mu \nu}$. The covariant derivative associated with the Levi-Civita connection will be denoted by $\nabla$ or just by ";". We will denote by $T_{...(\mu\nu)...}$ the symmetric part and by $T_{...[\mu\nu]...}$ the antisymmetric part of tensor $T_{...\mu\nu...}$ with respect to indices $\mu$ and $\nu$ (analogous symbols will be used for more indices).
\subsection{General properties of CYK tensors in four dimensions}
CYK tensor can be defined in general case -- for $p$-forms on $n$-dimensional manifold.
 However, here we restrict ourselves to an oriented manifold of dimension $n=4$ and by CYK tensor we mean a (two-index) antisymmetric tensor (two-form).\\ Let $Q_{\mu\nu}$ be an antisymmetric tensor field (two-form) on $M$ and by $\mathcal{Q}_{\lambda\mu\nu}$ let us denote  a (three-index) tensor defined as follows:
 \begin{equation}
 \mathcal{Q}_{\lambda \mu\nu}:=Q_{\lambda \mu;\nu}+Q_{\nu \mu;\lambda}-\frac{2}{3}(g_{\nu \lambda} {Q^{\rho}}_{\mu; \rho}+g_{\mu ( \lambda } {{Q_{\nu)}}^{\rho}}_{; \rho}) \label{cal_Q}
 \end{equation}
The object $\mathcal{Q}$ has the following algebraic properties: $\mathcal{Q}_{\lambda \mu\nu}g^{\lambda \nu}=0=\mathcal{Q}_{\lambda \mu\nu} g^{\lambda \mu}$, $\mathcal{Q}_{\lambda \mu\nu}=\mathcal{Q}_{\nu \mu\lambda}$, i.e. it is traceless and partially symmetric.
\begin{df}
	An antisymmetric tensor $Q_{\mu\nu}$ is a conformal Yano--Killing tensor (or simply CYK tensor) for the metric $g$ iff $\mathcal{Q}_{\lambda \mu\nu}(Q,g)=0$.
\end{df}
In other words, $Q_{\mu\nu}$ is a conformal Yano--Killing tensor if it fulfils the following equation:
\begin{equation}
Q_{\lambda \mu;\nu}+Q_{\nu \mu;\lambda}=\frac{2}{3}(g_{\nu \lambda} {Q^{\rho}}_{\mu; \rho}+g_{\mu ( \lambda } {{Q_{\nu)}}^{\rho}}_{; \rho}) \label{df_Q}
\end{equation}
(first proposed by Tachibana and Kashiwada \cite{Tachi-Kashi}).

\paragraph{Hodge duality} In the space of differential forms on an oriented manifold one can define a mapping called the Hodge duality (Hodge star). It assigns to every $p$-form an $(n-p)$-form (where $n$ is the dimension of the manifold). We consider the case of $n=4$ and $p=2$. The Hodge star then becomes a mapping which assigns to a two-form $\omega$ a two-form $\ast \omega$. 
We can express this mapping in the following way:
\begin{equation}
 \ast \omega_{\alpha \beta}=\frac{1}{2}{\varepsilon_{\alpha \beta}}^{\mu \nu}\omega_{\mu \nu} \label{Hodge}
\end{equation}
where $\varepsilon_{\alpha \beta \gamma \delta}$ is the antisymmetric Levi-Civita tensor\footnote{It can be defined by the formula $\varepsilon_{\alpha \beta \gamma \delta}=\sqrt{-\det g_{\mu\nu}}\epsilon_{\alpha \beta \gamma \delta}$, where
	\begin{equation}
	\epsilon_{\alpha \beta \gamma \delta}=\left \{\begin{array}{l}
	+1 \quad \mbox{if $\alpha \beta \gamma \delta$ is an even permutation of 0, 1, 2, 3}\\
	-1 \quad \mbox{if $\alpha \beta \gamma \delta$ is an odd permutation of 0, 1, 2, 3}\\
	\hspace{0.25 cm}0 \hspace{0.2 cm} \mbox{ in any other case}
	\end{array} \right. \nonumber
	\end{equation}} determining orientation of the manifold ($\frac{1}{4!} \varepsilon_{\alpha \beta \gamma \delta} \mathrm{d} x^{\alpha} \wedge  \mathrm{d} x^{\beta} \wedge \mathrm{d} x^{\gamma} \wedge \mathrm{d} x^{\delta}$ is the volume form of the manifold $M$). For the Lorentzian metric we have $\ast \ast \omega=- \omega$. Due to CYK tensor being a two-form, it is reasonable to ask what are the properties of its dual. Let $Q$  be a CYK tensor and $\ast Q$ its dual. Moreover, let us introduce the following covector $\chi_{\mu}:=\nabla^{\nu} \ast Q_{\nu \mu}$. It was proved in \cite{jez-luk}, that
\begin{equation}
\ast Q_{\lambda \mu;\nu}+\ast Q_{\nu \mu;\lambda}=\frac{2}{3}\left(g_{\nu \lambda} \chi_{\mu}-g_{\mu ( \lambda }\chi_{\nu)}\right) \label{ast_Q}
\end{equation}
It is not hard to recognize that this is Eq. (\ref{df_Q}) for the tensor $\ast Q$. It proves the following theorem:
\begin{tw}
	Let $g_{\mu\nu}$ be a metric tensor on  a four-dimensional differential manifold $M$. An antisymmetric tensor $Q_{\mu \nu}$ is a CYK tensor of the metric $g_{\mu\nu}$ if and only if its dual $\ast Q_{\mu \nu}$ is also a CYK tensor of this metric. \label{dualCYK}	
\end{tw}
The above theorem implies that for every four-dimensional manifold, solutions of Eq. (\ref{df_Q}) exist in pairs -- to each solution we can assign the dual solution (in the Hodge duality sense). For any tensor field $T_{\mu\nu}$ on $M$ holds:
\begin{equation}
T_{\lambda \kappa; \nu \mu}-T_{\lambda \kappa; \mu \nu}=T_{\sigma \kappa} {R^{\sigma}}_{\lambda \nu \mu}+T_{\lambda \sigma} {R^{\sigma}}_{\kappa \nu \mu} \label{comT}
\end{equation}
The above identity follows directly from the Riemann tensor definition.
\paragraph{Wave equation satisfied by CYK tensor} Let $M$ be a four-dimensional manifold equipped with a metric tensor $g_{\mu\nu}$. If $Q_{\lambda \sigma}$ and ${R^{\sigma}}_{\lambda \nu \mu}$ are a CYK tensor and the Riemann tensor respectively then the following equality holds:
\begin{equation}
\nabla_{\mu} \nabla^{\mu} Q_{\lambda \kappa}={R^{\sigma}}_{\kappa \lambda \nu} {Q_{\sigma}}^{\nu}-R_{\sigma[\kappa}{Q_{\lambda]}}^{\sigma} \label{dalamb_Q}
\end{equation} \label{ddQ}
The above equality was proved in 
\cite{jez-luk}. For clarity of the exposition, we present the proof of the above equation in the appendix \ref{shifted_proof}.

\subsection{The relation between CYK tensor and Maxwell field}
We use vacuum Maxwell equations in terms of Maxwell field. Maxwell field $F_{\mu \nu}$ is a two-form (antisymmetric tensor) field. The vacuum Maxwell equations in terms of Maxwell field take the form
\begin{equation}
\left\{\begin{array}{r}
\mathrm{d} F=0\\[0.5ex]
\mathrm{d} \ast \!\!F=0
\end{array}   \right. \iff
\left\{\begin{array}{r}
F_{[\mu\nu;\lambda]}=0 \\[0.5ex]	{F_{\mu\nu}}^{;\mu}=0			
\end{array}   \right. \label{eqMaxwell}
\end{equation}
 where $\ast$ denotes Hodge duality (\ref{Hodge}). We consider real Maxwell fields. Taking a divergence of the first Maxwell equation (\ref{eqMaxwell}) and combining with (\ref{comT}), we can easily transform the d'Alembertian of Maxwell field:
 \begin{eqnarray}
	\Box F_{\mu \nu}&=&{F_{\mu \nu ;\lambda}}^{;\lambda} \nonumber \\
	&=&{F_{\lambda \nu ;\mu}}^{;\lambda}- {F_{\lambda \mu ;\nu}}^{;\lambda} \nonumber \\
	&=& \underbrace{{{F_{\lambda \nu}}^{;\lambda}}_{\mu}}_{0}-g^{\rho \lambda}({R^{\alpha}}_{\lambda \rho \mu} F_{\alpha \nu}+{R^{\alpha}}_{\nu \rho \mu} F_{\lambda \alpha}) \nonumber \\
	&& - \underbrace{{{F_{\lambda \mu}}^{;\lambda}}_{\nu}}_{0}+g^{\rho \lambda}({R^{\alpha}}_{\lambda \rho \nu} F_{\alpha \mu}+{R^{\alpha}}_{\mu \rho \nu} F_{\lambda \alpha}) \nonumber \\
	&=& g^{\rho \lambda} \left[ ({R^{\alpha}}_{\lambda \rho \nu} F_{\alpha \mu}-{R^{\alpha}}_{\lambda \rho \mu} F_{\alpha \nu})+({R^{\alpha}}_{\mu \rho \nu} F_{\lambda \alpha}-{R^{\alpha}}_{\nu \rho \mu} F_{\lambda \alpha}) \right] \nonumber \\
	&=& ({R^{\lambda}}_{\alpha\lambda \nu} {F_{\mu}}^{\alpha}-{R^{\lambda}}_{\alpha\lambda \mu} {F_{\nu}}^{\alpha})+({R^{\alpha}}_{\nu \mu \lambda} {F^{\lambda}}_{\alpha}-{R^{\alpha}}_{\mu \nu \lambda} {F^{\lambda}}_{\alpha}) \nonumber
 \end{eqnarray}
 and we obtain
 \begin{equation}
\Box F_{\mu \nu}=-2 {R^{\lambda}}_{\alpha\lambda[\mu} {F_{\nu]}}^{\alpha} -2 {R^{\alpha}}_{\left[\mu \nu \right] \lambda} {F^{\lambda}}_{\alpha} \label{ddF}
 \end{equation}
Combining Maxwell equations (\ref{eqMaxwell}) with CYK equation (\ref{df_Q}), we can check that the term $\nabla^{\lambda}F_{\mu \nu}\nabla_{\lambda} Q^{\mu \nu}$ is vanishing. More precisely,
\begin{eqnarray}
0&=&\overbrace{F_{[\mu \nu;\lambda]}}^{0}Q^{\mu\nu; \lambda} \nonumber \\
&=&F_{\mu \nu; \lambda}Q^{\mu \nu;\lambda}+ 2 F_{\lambda \mu;\nu}Q^{\mu \nu;\lambda} \nonumber \\
&=&3 F_{\mu \nu;\lambda} Q^{\mu \nu;\lambda}+2 F_{\lambda \mu;\nu}(Q^{\mu \nu;\lambda}+ Q^{\mu \lambda;\nu}) \nonumber \\
&=&2 F_{\lambda \mu;\nu}\underbrace{ \left[Q^{ \mu \nu;\lambda}+Q^{\mu \lambda; \nu}+\frac{2}{3}(g^{\nu \lambda} {Q^{\rho\mu}}_{; \rho}+g^{\mu ( \lambda } {Q^{\nu)\rho}}_{; \rho})\right]}_{ \mathcal{Q}^{\mu\nu \lambda}=0} \nonumber \\
&&+3F_{\mu \nu;\lambda}Q^{\mu \nu; \lambda} - 2 \underbrace{{F^{\lambda}}_{\mu;\lambda}}_{0}{Q^{\rho \mu}}_{; \rho} \nonumber \\
&=&3F_{\mu \nu;\lambda}Q^{\mu \nu; \lambda} \label{dFdQ}
\end{eqnarray}
 Finally,	the d'Alembertian of Maxwell--CYK contraction takes the following form:
 \begin{eqnarray}
\Box(F_{\mu \nu}Q^{\mu \nu})&=&Q^{\mu \nu}\Box F_{\mu \nu}+F^{\mu \nu}\Box Q_{\mu \nu}+2 F_{\mu \nu;\lambda}Q^{\mu \nu; \lambda} \nonumber \\
&=&Q^{\mu \nu}\Box F_{\mu \nu}+F^{\mu \nu}\Box Q_{\mu \nu}	 \label{dalamb_FQ}
 \end{eqnarray}
 The last equality is implied by (\ref{dFdQ}). The above considerations lead to the following 
 \begin{tw} Let $F_{\mu\nu}$, $Q_{\mu\nu}$ and ${R^{\sigma}}_{\lambda \nu \mu}$ be respectively a Maxwell field, a CYK tensor and the Riemann tensor corresponding to the metric $g_{\mu\nu}$. Then
 	\begin{equation}
	 \Box(F_{\mu \nu} Q^{\mu \nu})+\frac{1}{2}F^{\sigma \lambda}R_{\sigma \lambda \mu \nu}Q^{\mu \nu}+Q^{\mu \nu}R_{\sigma \mu}F_{\nu}{^{\sigma}}=0\label{FQ_Riemann}
 	\end{equation} \label{General_thm}
 \end{tw}
 \begin{proof}
 	 Making use of equations (\ref{dalamb_Q}), (\ref{ddF}) and (\ref{dalamb_FQ}), we can transform Maxwell--CYK contraction in the following way:
 	 \begin{eqnarray}
 	 \Box(F_{\mu \nu}Q^{\mu \nu})&=&-2 Q^{\mu \nu}\left({R^{\lambda}}_{\sigma\lambda\mu} {F_{\nu}}^{\sigma} +{R^{\sigma}}_{\mu \nu \lambda} {F^{\lambda}}_{\sigma}\right)+F^{\mu \nu} \left({R^{\sigma}}_{\nu \mu \lambda} {Q_{\sigma}}^{\lambda}+R_{\sigma\mu}{Q_{\nu}}^{\sigma} \right) \nonumber\\
 	 &=&-Q^{\mu \nu} R_{\sigma \mu}{F_{\nu}}^{\sigma}+R_{\sigma \mu \nu \lambda}(2 Q^{\mu\nu} F^{\sigma \lambda}-F^{\mu\nu} Q^{\sigma \lambda})\nonumber \\
 	 &=&-\frac{1}{2}F^{\sigma \lambda}R_{\sigma \lambda \mu \nu}Q^{\mu \nu}-Q^{\mu \nu}R_{\sigma \mu} F_{\nu}{^{\sigma}} \nonumber
 	 \end{eqnarray}
 		(the last equality uses Bianchi identity $R^{\sigma}{_{[\mu \nu \lambda]}}=0$).
 \end{proof}
 		It is convenient to split the Riemann tensor into Weyl tensor $C_{\sigma \lambda \mu \nu}$, Ricci tensor $R_{\mu \nu}$ and curvature scalar $R$, and rewrite Eq. (\ref{FQ_Riemann}) in the equivalent form
\begin{equation}
\left(\Box-\frac{1}{6}R \right)(F_{\mu \nu}Q^{\mu \nu})+\frac{1}{2}F^{\sigma \lambda}C_{\sigma \lambda \mu \nu}Q^{\mu \nu}=0 \label{FQ_Weyl}
\end{equation}
 We may note here that the above equation (\ref{FQ_Weyl}) is crucial for our further investigation. For flat spacetime Eq. (\ref{FQ_Weyl}) reduces to the wave equation $\displaystyle \Box \phi=0$, where $\phi:=F_{\mu \nu}Q^{\mu \nu}$ is a scalar function.
  We will show in the next section that for the Kerr case we can reduce Eq. (\ref{FQ_Weyl}) to Fackerell--Ipser equation. 
  Let us remind: a Weyl tensor ${C^{\sigma}}_{\lambda \mu \nu}$ and a Maxwell field $F_{\mu\nu}$ remain unchanged under a conformal rescaling $g_{\mu\nu} \to \Omega^{2} g_{\mu\nu}$ for any positive function $\Omega$ on $M$. Moreover, tensor $\mathcal{Q}_{\lambda \mu\nu}$  (see (\ref{cal_Q})) transforms under the conformal rescaling in the following way:
  \begin{equation}
   \mathcal{Q}_{\lambda \mu\nu}(Q,g)=\Omega^{-3} \mathcal{Q}_{\lambda \mu\nu}(\Omega^{3}Q,\Omega^{2}g)  \nonumber
  \end{equation}
which implies 
\begin{st}
	If $Q_{\mu \nu}$ is a CYK tensor for the metric $g_{\mu \nu}$, then $\Omega^{3} Q_{\mu\nu}$ is a CYK tensor for the conformally rescaled metric  $\Omega^{2} g_{\mu \nu}$.
\end{st}
Moreover, the (upper index) tensor $Q^{\alpha\beta}=g^{\alpha\mu}g^{\beta\nu}Q_{\mu\nu}$ rescales by $\Omega^{-1}$.
In addition to this, let $\phi$ be a scalar function on four-dimensional manifold. If $\phi$ rescales conformally $\tilde{\phi} \to \Omega^{-1} \phi$, then the operator presented below transforms under conformal change of a metric ($\widetilde{g}_{\mu\nu}=\Omega^{2} g_{\mu\nu}$) in the following way:
\begin{equation}
\left( \widetilde{\square}-\frac{1}{6} \widetilde{R} \right)\widetilde{\phi}=\Omega^{-3}\left( \Box-\frac{1}{6} R \right) \phi \label{conf_inv_op}
\end{equation}
where $R$ is a curvature scalar.

The above facts lead to a proposition presented below.
\begin{st}
	The equation (\ref{FQ_Weyl}) remains unchanged under conformal transformation of the metric: ${g} \to
 \widetilde{g}=\Omega^2 g$.
\end{st}

\section{Electrodynamics on Kerr background}
In this section we consider Eq. (\ref{FQ_Weyl}) for Kerr black hole metric. There is only one pair of solutions (\ref{cyk_kerr}) of CYK equation (\ref{df_Q}) known in the literature. It turns out that this pair of CYK tensors (\ref{cyk_kerr}) possesses some additional properties (see (\ref{diag_Y})) which enable one to obtain a single scalar electromagnetic wave equation describing the evolution of the Maxwell field.
\subsection{Kerr spacetime}
Locally, the Kerr solution to the vacuum Einstein equations is the metric $g_{\mu\nu}$ which in Boyer--Lindquist
coordinates takes the form
\begin{eqnarray}
\label{kerr.metric}
\nonumber
g_{\mu\nu}\mathrm{d}x^\mu \mathrm{d}x^\nu &=& \rho^2\left(\frac{1}{\Delta} \mathrm{d}r^2+ \mathrm{d}\theta^2\right)
+\frac{\sin^2\theta}{\rho^{2}} \left(a \mathrm{d}t-(r^2+a^2)\mathrm{d}\varphi\right)^2  \\
&\quad& -  \frac{\Delta}{\rho^{2}}\left(\mathrm{d}t-a \sin^2 \theta \, \mathrm{d}\varphi\right)^2 
\end{eqnarray}
where
\begin{eqnarray}
\label{rho.kerr}
\rho^2 &=& r^2+a^2\cos^2 \theta \vphantom{\frac11} 
\\
\label{delta.kerr}
\Delta  &=& (r^2+a^2) -2 m r 
\end{eqnarray}
with $t \in \mathbb{R}$, $r \in \mathbb{R}$, and $\theta$, $\varphi$ being the standard coordinates parameterizing a
two-dimensional sphere. We will keep away from zeros of $|\rho|$ and $\Delta$, and ignore the coordinate singularities $\sin \theta=0$. This metric describes a rotating object of mass $m$ and angular momentum $J=ma$. The advantage of the above coordinates is that for $r$ much grater than $m$ and $a$ the metric  becomes asymptotically flat, i.e. $g \approx -\mathrm{d} t^2+\mathrm{d} r^2 +r^2 (\mathrm{d} \theta^2+\sin^2 \theta \mathrm{d} \phi^2)$.

\subsection{Properties of CYK tensor for Kerr spacetime \label{section_CYK_Kerr}}
Finding a solution of  the CYK tensor equation (\ref{df_Q}) for Kerr is not an easy task. We are dealing with a quite  complicated overdetermined system of differential equations for components of $Q_{\mu \nu}$. However, there is one solution known in the literature (see \cite{gibons}). Aksteiner and Anderson have shown \cite{AndAks} that 
 only one pair of solutions exists for type D spacetimes. According to the theorem \ref{dualCYK}, this solution has its dual companion. 
We denote them by $Y:=Q_{\mathrm{Kerr}}$ and $\ast Y:=\ast Q_{\mathrm{Kerr}}$:
\begin{eqnarray}
Y&=&r \sin \theta \mathrm{d} \theta \wedge \left[(r^2 +a^2) \mathrm{d} \varphi - a \mathrm{d} t \right] + a \cos \theta \mathrm{d} r \wedge (\mathrm{d} t -a \sin^2 \theta \mathrm{d} \varphi) \nonumber \\
	\ast Y&=&a \cos \theta \sin \theta \mathrm{d} \theta \wedge \left[(r^2 +a^2) \mathrm{d} \varphi - a \mathrm{d} t \right]+r \mathrm{d} r \wedge \left( a \sin ^{2} \theta \mathrm{d} \varphi - \mathrm{d} t \right)  \label{cyk_kerr}
\end{eqnarray}
{
The CYK solutions for Kerr spacetime (\ref{cyk_kerr}) are given in the explicit coordinate system which is not global. We roughly discuss the global existence of (\ref{cyk_kerr}). There are two points which my be problematic with the analytic extension of (\ref{cyk_kerr}):
	\begin{enumerate}
		\item The solution is not global on a whole sphere $r=\mathrm{const}$ --- It is a well-known problem how the formulae (like $\sin\theta {\mathrm d}\theta \wedge {\mathrm d}\varphi$) can be extended on the whole sphere which is not covered by $(\theta,\varphi)$ coordinates. A suitable change of the coordinates is needed.
		\item Analytical extension through the horizon (together with the corresponding coordinate system) --- We have transformed the CYK tensors (\ref{cyk_kerr}) into advanced Eddington--Finkelstein coordinates which are well-defined on the horizon. They are smooth on the horizon and can be extended analytically through it.
	\end{enumerate}}
The Riemann tensor $R_{\alpha \beta \rho \sigma}$ of (Ricci flat) Kerr spacetime is equal to its Weyl tensor. Weyl curvature tensor has two pairs of antisymmetric indices. The algebraic structure of the Weyl tensor
$C_{\mu \nu}{^{\lambda \kappa}}$ allows it to be treated as an endomorphism in the space of two-forms at each point $p\in M$:
\[ C:\operatornamewithlimits{\bigwedge}^{2} \mathrm{T}^{*}_p{\mathrm{M}} \rightarrow \operatornamewithlimits{\bigwedge}^{2} \mathrm{T}^{*}_p{\mathrm{M}} \]
In the six-dimensional space $\displaystyle\operatornamewithlimits{\bigwedge}^{2} \mathrm{T}^{*}_p{\mathrm{M}}$ we can distinguish a two-dimensional subspace $\mathbb{V}$ which is spanned by $Y$ and $\ast Y$. $\mathbb{V}$
proves 
to be an invariant subspace of the endomorphism $C$. More precisely,
\begin{eqnarray}
Y^{\lambda \kappa} C_{\mu \nu \lambda \kappa } \mathrm{d} x^{\mu} \wedge \mathrm{d} x^{\nu}&=&\frac{4 m}{\rho^4} \left\{[r^{2}-a^{2} \cos^{2} \theta] \sin \theta \mathrm{d} \theta \wedge [(r^{2}+a^{2}) \mathrm{d} \varphi- a \mathrm{d} t] \right. \nonumber \\
& & +\left. 2 a r \cos \theta \mathrm{d} r \wedge [a \sin^{2} \theta \mathrm{d} \theta - \mathrm{d} t] \right\} \nonumber \\
\ast Y^{\lambda \kappa} C_{\mu \nu \lambda \kappa }  \mathrm{d} x^{\mu} \wedge \mathrm{d} x^{\nu}&=&\frac{4 m}{\rho^{4}} \left\{2 a r \sin \theta \cos \theta \mathrm{d} \theta \wedge [a \mathrm{d} t- (r^{2}+a^{2}) \mathrm{d} \varphi] \right. \nonumber \\
& & +\left. [r^{2}-a^{2} \cos^{2} \theta] \mathrm{d} r \wedge (a \sin^{2} \theta \mathrm{d} \varphi -\mathrm{d} t) \right\} \label{YC_contraction}
\end{eqnarray}
We can test 
by direct computation whether 
the above result is a linear combination of $Y$ and $\ast Y$, and if the endomorphism $C$ restricted to $\mathbb{V}$ reduces to the simple matrix form:
\begin{equation}\label{HQ}
C \left[\begin{array}{c}
Y \\
\ast  Y
\end{array}\right]= \frac{4 m}{(r^2+a^2 \cos^{2} \theta)^3 }
\, {\bf H} \left[ \begin{array}{c}
Y \\
\ast Y
\end{array} \right]
\end{equation}
where
$\displaystyle {\bf H}:= 
\left[ \begin{array}{cc}
	r (r^2-3 a^2 \cos^2 \theta ) & (3 r^2 - a^2 \cos^2 \theta ) a \cos \theta \\
	-(3 r^2 - a^2 \cos^2 \theta ) a \cos \theta & r (r^2-3 a^2 \cos^2 \theta )
	\end{array} \right] 
$.

Eigenfunctions of the real matrix $\textbf{H}$ are complex:
\begin{equation}
\lambda=(r \pm \imath a \cos \theta )^3
\end{equation}
and the corresponding eigenvectors being:
\begin{equation}
\left\{\left[\begin{array}{c}
-\imath \\
1
\end{array} \right]\begin{array}{c}
\\
,
\end{array}  \left[\begin{array}{c}
\imath \\
1
\end{array} \right] \right\}
\end{equation}
\paragraph{Fact:} The two-form $Y - \imath \ast Y$ diagonalizes the Weyl endomorphism $C$. More precisely, we have
\begin{eqnarray}
C_{\mu \nu}{^{\lambda \kappa}} Y_{\lambda \kappa} - \imath C_{\mu \nu}{^{\lambda \kappa }} (\ast Y_{\lambda \kappa}) =2 V \left( Y_{\mu \nu} - \imath \ast\! Y_{\mu \nu} \right) \label{diag_Y}
\end{eqnarray}
where the eigenfunction is
\begin{equation}
V=\frac{2 m}{(r- \imath a \cos \theta)^3 }
\end{equation}
The above fact has serious consequences in the further theory formulation. In the spacetime where (\ref{diag_Y}) holds, a scalar electromagnetic wave equation can be introduced. This will be presented in detail in the next paragraph.
\subsection{Scalar electromagnetic wave equation in Kerr spacetime \label{section_FI_Kerr}}
Now we return to the equation (\ref{FQ_Weyl}) and rewrite it for $Y_{\mu\nu}$ and its dual $\ast Y_{\mu\nu}$ multiplied by $\imath$:
\begin{equation}
\left\{ \begin{array}{rcl} \left( \Box-\frac{1}{6}R \right)(F_{\mu \nu}Y^{\mu \nu})+\frac{1}{2}F^{\sigma \lambda}C_{\sigma \lambda \mu \nu}Y^{\mu \nu}&=&0 \\
\left( \Box -\frac{1}{6}R \right)( \imath F_{\mu \nu}  (\ast Y^{\mu \nu}))+\frac{\imath}{2}F^{\sigma \lambda}C_{\sigma \lambda \mu \nu}(\ast Y^{\mu \nu})&=&0
\end{array} \right. \label{uklad}
\end{equation}
For Kerr metric, as a solution of vacuum Einstein equations, the curvature scalar $R$ vanishes. 
Subtracting both sides and using (\ref{diag_Y}), we obtain:
\begin{equation}
 \Box\left[F_{\mu \nu} \left( Y^{\mu \nu}- \imath(\ast Y^{\mu \nu}) \right)\right]+V F^{\mu \nu} \left(Y_{\mu \nu}- \imath(\ast Y_{\mu \nu})  \right)=0
\end{equation}
Introducing \begin{equation}\label{defPhi} 
\Phi:=\frac{\imath}{2}F^{\mu \nu} \left[Y_{\mu \nu}- \imath(\ast Y_{\mu \nu})\right] \, ,
\end{equation}
we get a scalar electromagnetic wave equation:
\begin{equation}
\Box \Phi + V \Phi=0 \nonumber
\end{equation}
The above calculations prove the following 
\begin{tw}[Fackerell--Ipser]
	Dynamics of a Maxwell field in the Kerr spacetime can be reduced to the scalar wave equation:
	\begin{equation}
	\Box \Phi + V \Phi=0 \label{FI}
	\end{equation}
	where $\displaystyle \Phi=\frac{\imath}{2} F^{\mu \nu} \left[Y_{\mu \nu}- \imath(\ast Y_{\mu \nu})\right]$,
   $\displaystyle V=\frac{2 m}{(r- \imath a \cos \theta)^3} $.
\end{tw}
\subsection{Other approaches to Fackerell--Ipser equation}
The F--I equation was derived at the beginning of 1970s \cite{FackIps} using Newman--Penrose formalism.  The formulation of Maxwell equations on Kerr background in terms of Newman--Penrose formalism and Teukolsky functions are discussed by Chandrasekhar \cite{Chandra}. In the Newman--Penrose formalism electromagnetic tensor $F$ is characterized by three scalars:
\begin{equation}
\phi_{+1}=F_{\mu \nu}l^{\mu} m^{\nu} \quad \phi_{0}=\frac{1}{2}F_{\mu \nu} (l^{\mu} n^{\nu}+\bar{m}^{\mu}m^{\nu}) \quad \phi_{-1}=F_{\mu \nu} \bar{m}^{\mu} n^{\nu} \label{NP}
\end{equation}
where null tetrad $l^{\mu}, n^{\nu}, m^{\rho},\bar{m}^{\sigma}$ is defined at each point of spacetime; the vectors $l^{\mu}$ and $n^{\nu}$ are real while $m^{\rho}$ and $\bar{m}^{\sigma}$ are complex conjugates of one another; moreover $m_{\rho}\bar{m}^{\rho}=-1$,  $l_{\mu} n^{\mu}=1$. All other scalar products vanish. The normalization of the null tetrad is invariant under the action of the six-dimensional group of Lorentz transformations. For further calculations, we will use Carter tetrad:
\begin{eqnarray}
l&=&\frac{1}{\sqrt{2 \rho^2}}\left[\sqrt{\Delta} \mathrm{d} t-\frac{\rho^2}{\sqrt{\Delta}} \mathrm{d} r - a \sqrt{\Delta} \sin^2 \theta \mathrm{d} \varphi \right] \label{Carter_tetrad}\\
n&=&\frac{1}{\sqrt{2 \rho^2}}\left[\sqrt{\Delta} \mathrm{d} t+\frac{\rho^2}{\sqrt{\Delta}} \mathrm{d} r - a \sqrt{\Delta} \sin^2 \theta \mathrm{d} \varphi \right] \nonumber \\
m&=&\frac{1}{\sqrt{2 \rho^2}} \left[ \imath a \sin \theta \mathrm{d} t - \rho^2 \mathrm{d} \theta - \imath (r^2+a^2) \sin \theta \mathrm{d} \varphi \right] \nonumber \\
\bar{m}&=&\frac{1}{\sqrt{2 \rho^2}} \left[-\imath a \sin \theta \mathrm{d} t - \rho^2 \mathrm{d} \theta + \imath (r^2+a^2) \sin \theta \mathrm{d} \varphi \right ] \nonumber
\end{eqnarray}
The solution $\Phi$ of F--I equation (\ref{FI}) is related to Newman--Penrose electromagnetic scalar $\phi_{0}$ from (\ref{NP}) by the following formula
\begin{equation}
\Phi=(r-\imath a \cos \theta) \phi_{0} \label{FI_NP}
\end{equation}
{
The above equation 
is an algebraic relation between $\Phi$ and $\phi_{0}$. Eq. (\ref{FI}) with the help of Eq. (\ref{FI_NP}) gives an explicit second-order equation for $\phi_{0}$. It was originally obtained by Fackerell and Ipser in \cite{FackIps}.}\\
CYK tensors for Kerr spacetime can be easily expressed in terms of Carter null tetrad. It 
simplifies equation (\ref{FI_NP}).  CYK tensors (\ref{cyk_kerr}) in Carter tetrad (\ref{Carter_tetrad}) have the form
\begin{eqnarray}
Y&=& 2 a \cos \theta n \wedge l- 2 \imath r \bar{m} \wedge m \\
\ast Y&=&  2 r n \wedge l+ 2 \imath a \cos \theta \bar{m} \wedge m
\end{eqnarray}
In particular, $ \imath (Y-\imath \ast Y)=2(r-\imath a \cos \theta)(n \wedge l +\bar{m} \wedge m)$. For further details see chapter 2.5.1 in \cite{Aksteiner_phd}.
The F--I equation is investigated with the use of spinorial approach. For the recent results see \cite{And_spinor}.
\section{Miscellaneous results \label{misc_result}}
\subsection{Generalization for de Sitter spacetimes}
Equation (\ref{FQ_Weyl}) remains true for four-dimensional spacetimes equipped with CYK tensor. In particular, it is valid for some metrics which are solutions to Einstein equations with cosmological constant $\Lambda$. Only the existence of CYK tensor is needed for our construction. We present an example of generalized F--I equation for Kerr--de Sitter spacetime.
\paragraph{Fackerell--Ipser equation in Kerr--de Sitter spacetime.} The Kerr solution (\ref{kerr.metric}) can be further generalized to include a non-zero cosmological constant $\Lambda$. In the following section we will not distinguish the sign of $\Lambda$ and refer to Kerr--de Sitter
and Kerr--anti--de Sitter spacetimes 
as KdS spacetime.  Locally, in Boyer--Lindquist coordinates (see \cite{GibHaw}), the metric takes the form
\begin{eqnarray}
	\label{kds.metric}
	\nonumber
	g_{\mu\nu}\mathrm{d}x^\mu \mathrm{d}x^\nu &=& \rho^2\left(\frac{1}{\Delta_r} \mathrm{d}r^2+\frac{1}{\Delta_\theta} \mathrm{d}\theta^2\right)
	+\frac{\sin^2\theta\Delta_\theta}{\rho^{2}\Xi^{2}} \left(a \mathrm{d}t-(r^2+a^2)\mathrm{d}\varphi\right)^2  \\
	&\quad& -  \frac{\Delta_r}{\rho^{2}\Xi^{2}}\left(\mathrm{d}t-a \sin^2 \theta \, \mathrm{d}\varphi\right)^2 
\end{eqnarray}
where
\begin{eqnarray}
	\label{rho.2}
	\rho^2 &=& r^2+a^2\cos^2 \theta \vphantom{\frac11} 
\\
	\label{delta.r}
	\Delta_r  &=& (r^2+a^2)\left(1-\frac{\Lambda}3 r^2\right) -2 m r 
\\
	\label{delta.theta}
	\Delta_\theta &=& 1+\frac{a^2 \Lambda}3  \cos^2 \theta 
\\
	\label{xi}
	\Xi &=& 1+\frac{a^2\Lambda}3 
\end{eqnarray}
with $t \in \mathbb{R}$, $r \in \mathbb{R}$, and $\theta$, $\varphi$ being the standard coordinates parameterizing the sphere.
We will keep away from zeros of $\rho$ and $\Delta_r$, and ignore the coordinate singularities $\sin \theta=0$. In the context of our construction, a natural question arises: what are the solutions of Eq. (\ref{df_Q}) for the KdS metric? It turns out that the solution (\ref{cyk_kerr}) for Kerr spacetime  generalizes in a simple way: $Y_{\mu \nu}$ and $\ast Y_{\mu \nu}$ (\ref{cyk_kerr}) are also the solutions of Eq. (\ref{df_Q}) for KdS metric (see \cite{Kubiz})
\begin{equation}
\left\{\begin{array}{r}
Y=Q_{\mathrm{Kerr}}=Q_{\mathrm{KdS}}\\
\ast Y=\ast Q_{\mathrm{Kerr}}=\ast Q_{\mathrm{KdS}}
\end{array}  \right.
\end{equation}
and $\ast Y$ is also a dual companion (in Hodge sense (\ref{Hodge})) for KdS metric.
 We would like to stress that the CYK two-forms for KdS spacetime do not depend on cosmological constant $\Lambda$.
 Surprisingly, it turns out that the CYK two-forms (\ref{cyk_kerr}) are solutions of CYK Eq. (\ref{df_Q}) for Kerr and for KdS spacetime.\\
 The Weyl tensor of KdS spacetime depends on $\Lambda$ in a non-trivial way. But one can check that CYK--Weyl contractions $Y^{\lambda \kappa} C_{\mu \nu \lambda \kappa }$ and $ \ast Y^{\lambda \kappa} C_{\mu \nu \lambda \kappa }$ do not depend on $\Lambda$ and the relation (\ref{YC_contraction}) holds also in KdS spacetime. We analyzed
 the problem with help of the
 symbolic software WATERLOO MAPLE to check the result (\ref{YC_contraction}) for KdS metric. The reasoning in subsections \ref{section_CYK_Kerr} and \ref{section_FI_Kerr} are based on equations (\ref{YC_contraction}) and (\ref{FQ_Weyl}). Moreover, the Weyl diagonalization process in subsection \ref{section_CYK_Kerr} is purely algebraical. It can be repeated for KdS spacetime in the same way. In subsection \ref{section_FI_Kerr} the equalities (\ref{uklad}) also hold for KdS spacetime. The difference between Kerr and KdS case is that the curvature scalar does not vanish. The curvature scalar for KdS metric is equal to $4 \Lambda$. We can rewrite Eq. (\ref{uklad}) in the following form:
 \begin{equation}
\left\{ \begin{array}{rcl} \left( \Box-\frac{2}{3} \Lambda \right)(F_{\mu \nu}Y^{\mu \nu})+\frac{1}{2}F^{\sigma \lambda}C_{\sigma \lambda \mu \nu}Y^{\mu \nu}&=&0 \\
\left( \Box -\frac{2}{3} \Lambda \right)( \imath F_{\mu \nu}  (\ast Y^{\mu \nu}))+\frac{\imath}{2}F^{\sigma \lambda}C_{\sigma \lambda \mu \nu}(\ast Y^{\mu \nu})&=&0
\end{array} \right. \label{uklad.KdS}
 \end{equation}
The Eq. (\ref{diag_Y}) holds true for Weyl tensor in KdS spacetime. Subtracting equations (\ref{uklad.KdS}) from each other and combining with (\ref{diag_Y}), we obtain
\begin{equation}
 \Box\left[F_{\mu \nu} \left( Y^{\mu \nu}- \imath(\ast Y^{\mu \nu}) \right)\right]+
 \left(V-\frac{2}{3} \Lambda\right) F^{\mu \nu} \left(Y_{\mu \nu}- \imath(\ast Y_{\mu \nu})  \right)=0
\end{equation}
Denoting by $\displaystyle \Phi=\frac{\imath}{2}F^{\mu \nu} \left[Y_{\mu \nu}- \imath(\ast Y_{\mu \nu})\right]$ and introducing $V_{\Lambda}=V-\frac{2}{3} \Lambda$, we prove the following 
\begin{tw}[Generalized Fackerell--Ipser equation]
	Dynamics of a Maxwell field in the Kerr--de Sitter spacetime can be reduced to the scalar wave equation
	\begin{equation}
	\Box \Phi + V_{\Lambda} \Phi=0 \label{FI.KdS}
	\end{equation}
	where $\displaystyle \Phi=\frac{\imath}{2} F^{\mu \nu} \left[Y_{\mu \nu}- \imath(\ast Y_{\mu \nu})\right]$,  $\displaystyle V_{\Lambda}=\frac{2 m}{(r- \imath a \cos \theta)^3}-\frac2{3} \Lambda $.
	\label{GFI}
\end{tw}
\subsection{Distorted Coulomb solution for Kerr--de Sitter spacetime \label{distorted_Culomb}}
Distorted Coulomb solution can be easily found by studying Kerr--Newman--de Sitter solution describing a rotating black hole with the electric charge (see \cite{GibHaw}). The Kerr--de Sitter metric (\ref{kds.metric}) is a special case of the Kerr--Newman--de Sitter metric for vanishing charge. The electromagnetic potential $A$ related to electromagnetism in Kerr--Newman--de Sitter spacetime is given by
\begin{equation}
A_{\mathrm{KdS}}=-\frac{q}{\Xi}\frac{r}{\rho^2}\left[\mathrm{d} t - a \sin^2 \theta \mathrm{d} \varphi \right] \label{A_KNdS}
\end{equation}
The associated Maxwell two-form ($F=\mathrm{d} A$) is given by
\begin{equation}
F_{\mathrm{KdS}}=\frac{q ( \rho^2-2 r^2)}{\rho^4 \Xi}\left[\mathrm{d} t - a \sin^2 \theta \mathrm{d} \varphi \right] \wedge \mathrm{d} r+\frac{q r a \sin 2 \theta }{\rho^4 \Xi}\left[a \mathrm{d} t - (a^2+r^2) \mathrm{d} \varphi\right] \wedge \mathrm{d} \theta \label{F_KNdS}
\end{equation}
and its dual companion (\ref{Hodge}) has the form
\begin{equation}
\ast F_{\mathrm{KdS}}= \frac{2 q r a \cos \theta }{\rho^4 \Xi} \left[\mathrm{d} t - a \sin^2 \theta \mathrm{d} \varphi \right] \wedge \mathrm{d} r - \frac{q ( \rho^2-2 r^2) \sin \theta }{\rho^4 \Xi} \left[a \mathrm{d} t - (a^2+r^2) \mathrm{d} \varphi\right] \wedge \mathrm{d} \theta
\end{equation}
The electric charge $e$ can be obtained from the Gauss law:
\begin{equation}
e=\frac{1}{4 \pi} \int_{\mathcal{S}^{2}} \ast F_{\mathrm{KdS}}=\frac{q}{\Xi} \label{charge}
\end{equation}
where $\mathcal{S}^2$ is a closed two-surface surrounding the horizon. We want to keep non vanishing $e$ (\ref{charge}) but simultaneously put parameter $q=0$ in the spacetime metric. This way we get the Maxwell solution (\ref{F_KNdS}) in Kerr--de Sitter spacetime.  The scalar $\Phi_{\mathrm{KdS}}=\frac{\imath}{2} F_{\mathrm{KdS}}^{\mu \nu} \left[Y_{\mu \nu}- \imath(\ast Y_{\mu \nu})\right]$  described by theorem  \ref{GFI} is given by
\begin{equation}
\Phi_{\mathrm{KdS}}=\frac{ - q \Xi}{r- \imath a \cos \theta} \label{Phi_KNdS}
\end{equation}

We want to point out a few facts:
\begin{enumerate}
	\item
		The CYK tensors for Kerr spacetime have encoded non trivial combination of symmetries. We would like to discuss it on the simplest example -- in the limit of Minkowski spacetime. CYK tensors in Minkowski spacetime form a twenty-dimensional space of solutions. It is a maximal dimension of the space of solutions. The solutions which span this twenty-dimensional space can be chosen in such a way that each of them is related to a particular symmetry (translation, rotation, etc.), see \cite{CYK_Minkowski} for further details. CYK tensors (\ref{cyk_kerr}) with $m=0$ are the non-trivial combinations of the simple symmetrical components. That means that the real and imaginary part of $\Phi$ cannot be treated like a projection of Maxwell field $F_{\mu \nu}$ on some well-known pure type of symmetries.
	\item	
		Even in the simple case when we have a monopole solution of Maxwell field ${}^{S} \! F$ in Eq. (\ref{F_KdS_decomp})  (${}^{S} \! F$ also fulfils Maxwell equations) the scalar field $\Phi$ has a higher multipole expansion. It comes from rich multipole structure of $Y$ and $\ast Y$ from Eqs (\ref{cyk_kerr}).
	\item 	
	The section \ref{distorted_Culomb} has also shown that the electric field of Kerr--Newman solution is not a Coulomb solution in the standard meaning. It contains, in addition to monopole part, also  higher order multipoles. Hence, we call the section ``Distorted Coulomb solution''.
	\item The solution (\ref{A_KNdS})  contains also a magnetic field, which is proportional to electric charge $e$ and rotation parameter $a$.  Even in the limit  $m = 0$, $\Lambda = 0$, corresponding to the Minkowski spacetime, the Maxwell field (\ref{F_KNdS}) is not spherically symmetric.
	
\end{enumerate}

 Let us consider the regime $m \to 0$, $\Lambda \to 0$. The Kerr--de Sitter metric (\ref{kds.metric}) reduces to the following form:
 \begin{equation}
\eta=-\mathrm{d} t^2+\frac{\rho^2}{r^2+a^2} \mathrm{d} r^2 + \rho^2 \mathrm{d} \theta^{2}+ (r^2+a^2) \sin^{2} \theta \mathrm{d} \varphi^2 \label{simpl_kds}
 \end{equation}
 corresponding to the Minkowski spacetime. Moreover, the spherical coordinates $(R(r,\theta),\Theta(r,\theta), \varphi)$ defined by the following formulae:
 \begin{eqnarray}
R&=&\sqrt{r^2+a^2\sin^2\theta} \nonumber \\
\sin\Theta&=&{\sin\theta \sqrt{1+\frac{a^2}{r^2}}\over
	\sqrt{1+\frac{a^2}{r^2}\sin^2\theta}} \label{RT_coords}\\
\cos\Theta &=&{\cos\theta \over
	\sqrt{1+\frac{a^2}{r^2}\sin^2\theta}} \nonumber
 \end{eqnarray}
transform the metric tensor (\ref{simpl_kds}) to the standard spherical form $\eta= -\mathrm{d} t^2 + \mathrm{d} R^2+ R^2 \mathrm{d} \Theta^2+R^2\sin^{2} \Theta \mathrm{d} \varphi^2$.

 The vector potential (\ref{A_KNdS}) (for $m = 0$, $\Lambda = 0$) in coordinates $(t,R,\Theta,\varphi)$ can be divided into spherically symmetric Coulomb term ${}^{S} \! A$  and higher rank multipole rest ${}^{R} \! A$:
 \begin{equation}
 A_{\mathrm{KdS}}={}^{S} \! A+{}^{R} \! A
 \end{equation}
where
 \begin{equation}
 {}^{S} \! A=-\frac{q}{R} \mathrm{d} t
 \end{equation}
The corresponding Maxwell tensor decomposition is
\begin{equation}
F_{\mathrm{KdS}}={}^{S} \! F+{}^{R} \! F \label{F_KdS_decomp}
\end{equation}
and the summands ${}^{S} \! F$ and ${}^{R} \! F $ separately fulfil Maxwell equations. A charge of ${}^{S} \! F$  obtained from the Gauss law is equal to $q=e$. The field ${}^{R} \! F$ has vanishing charge. The first terms of ${}^{R} \! A$ are related to magnetic and electric dipoles:
\begin{equation}
{}^{R} \! A=  \frac{q a \sin^2 \Theta}{R} \mathrm{d} \varphi-\frac{q a^2 (1-3 \cos^2 \Theta)}{2 R^3} \mathrm{d} t + \mathrm{l. o. t.}
\end{equation}

The multipole expansion of the scalar field $\Phi$ has a rich structure. The CYK tensors $Y$ and $\ast Y$ (\ref{cyk_kerr}) (for $m = 0$, $\Lambda = 0$) in spherical coordinates (\ref{RT_coords})  take the form
\begin{eqnarray}
Y & = & \nonumber
r\sin\theta \mathrm{d}\theta \wedge \left[ \left( r^2+a^2\right)\mathrm{d}\varphi - a\mathrm{d} t
\right]+ a\cos\theta \mathrm{d} r\wedge\left(\mathrm{d} t - a\sin^2\theta \mathrm{d}\varphi \right)
\\ \label{KerrYflat}
&= & \underbrace{R^3\sin\Theta\mathrm{d}\Theta\wedge\mathrm{d}\varphi}_{Y^{S}} +
\underbrace{a\mathrm{d}(R\cos\Theta)\wedge\mathrm{d} t}_{Y^{C}=a \mathrm{d} z \wedge \mathrm{d} t} \, .
\end{eqnarray}
\begin{eqnarray} \nonumber
\ast Y &=& a\cos\theta\sin\theta \mathrm{d}\theta \wedge \left[ \left(r^2+a^2
\right) \mathrm{d}\varphi - a\mathrm{d} t \right]+
r\mathrm{d} r \wedge \left(a\sin^2\theta \mathrm{d}\varphi -  \mathrm{d} t \right) \\ \label{Yflat}
&=& \underbrace{\mathrm{d} t\wedge R \mathrm{d} R}_{\ast Y^{S}} +
\underbrace{\frac12 a \mathrm{d}(R^2\sin^2\Theta) \wedge \mathrm{d} \varphi}_{\ast Y^{C}=a\mathrm{d} x \wedge \mathrm{d} y}
\end{eqnarray}
The first pair $(Y^{S}, \ast Y^{S})$ is the spherically symmetric CYK tensor with its dual companion and the second
pair $(Y^{C}, \ast Y^{C})$ is a constant correction which in Cartesian coordinates $(t,x,y,z)$ takes the simple form $Y^{C}=a \mathrm{d} z \wedge \mathrm{d} t$ and $\ast Y^{C}=a\mathrm{d} x \wedge \mathrm{d} y$. Hence, in the limit $m \to 0$, $\Lambda \to 0$ for the Kerr--de Sitter spacetime two-forms $Y^{S}, \ast Y^{S},Y^{C}, \ast Y^{C}$ are the CYK tensors for Minkowski spacetime. Using the $F_{\mathrm{KdS}}$ decomposition (\ref{F_KdS_decomp}) and $Y$, $\ast Y$ decompositions  (\ref{KerrYflat})--(\ref{Yflat}), the solution $\Phi_{\mathrm{KdS}}$ (\ref{FI.KdS}) can be divided into several parts. Introducing $(F|Y):=\frac1{2} F^{\mu \nu} Y_{\mu \nu}$, we have
\begin{eqnarray}
\Phi_{\mathrm{KdS}}&=&\imath \left( {}^{S} \!F+{}^{R} \! F\left|Y^{S}+Y^{C} +\frac{1}{\imath} [\ast Y^{S}+\ast Y^{C}] \right.\right) \nonumber \\
&=&\underbrace{\left({}^{S} \!F \left| \ast Y^{S} \right.\right)}_{\Phi^{S}}-\frac{1}{\imath}\underbrace{\left( {}^{S} \!F \left| Y^{S} \right.\right)}_{0}- 
\underbrace{\frac{1}{\imath} \left( {}^{S} \! F \left| Y^{C} - \imath  \ast Y^{C}\right. \right)- \imath\left( {}^{R} \! F \left| Y- \imath \ast Y \right.\right)}_{\Phi^{R}}
\end{eqnarray}
where
\begin{eqnarray*}
\Phi^{S}&=&\frac{q}{R}\\
\Phi^{R}&=&\frac{\imath q a \cos \Theta}{R^2}+\frac{q a^2 (1 - 3 \cos^2 \Theta)}{2 R^3} + O\left(\frac{1}{R^4}\right)
\end{eqnarray*}
 $\Phi^{S}$ and $\Phi^{R}$ both satisfy Eq. (\ref{FI.KdS}) in Minkowskian limit: $\Box \Phi=0$. The multipole structure of $\Phi^{R}$ has different interpretation than multipole expansion of Maxwell tensor\footnote{The multipole expansion of $\Phi^{R}$ is not in one to one correspondence to the electromagnetic multipoles. For example, the first term of $\Phi^{R}$ is given by $- \imath \left( {}^{S} \! F \left| Y^{C}\right. \right)-\imath \left( M \left| Y^{S}\right. \right)=(-1+2)\frac{\imath q a \cos \Theta}{R^2}$, where $M:=\mathrm{d} \left(\frac{q a \sin^2 \Theta}{R} \mathrm{d} \varphi\right)$ is a magnetic dipole part of ${}^{R} \! F$.} ${}^{R} \! F$. Similar considerations are made by Lynden-Bell (see \cite{Magic_EM}).
\subsection{Special singular solutions of Maxwell equations \label{Trautman_solutions}}
In this section we present a singular family of complex Maxwell fields on Kerr background. The solution $\Phi$ of F--I equation (\ref{FI}) constructed from any representant of this family is equal to zero.

All the information about an electromagnetic field can be encoded in a single, complex two-form
\begin{equation}
\mathcal{F}=F+\imath \ast F
\end{equation}
Let us recall that the Hodge star operator in (\ref{Hodge}) for Kerr metric satisfies $\ast^2=- \mathrm{id}$. A two-form $p$  is self dual in Hodge sense if
	\begin{equation}
	\ast p= \imath p \label{self_dual}
	\end{equation}
Note that $\mathcal{F}$ is self-dual. Maxwell equations in terms of $\mathcal{F}$ take a simple form
\begin{equation}
\mathrm{d} \mathcal{F}=0 \label{self_dual_Maxwell}
\end{equation}
 We will call  a two-form $p$ algebraically special if it fulfils
 \begin{equation}
 p \wedge p=0 \label{alg_special}
 \end{equation}
 Robinson and Trautman have proposed a singular, self-dual and algebraically special Maxwell field for optical geometry metrics (see \cite{trautman}). With the help of Paweł Nurowski \cite{Nurowski}, we have constructed such Maxwell field on Kerr. We will denote it by $\mathcal{F}$. $\mathcal{F}$ is built of two principal null covectors (\ref{Carter_tetrad}) and it has the following form:
\begin{equation}
\mathcal{F}=\frac{2 f(u- \imath a \cos \theta,\psi-\imath \mathrm{atanh}(\cos \theta))}{\imath \sqrt{\Delta} \sin \theta} n \wedge \bar{m} \label{solution_F}
\end{equation}
Equivalently in terms of coordinate forms
\begin{eqnarray}
\mathcal{F} &=& f(u- \imath a \cos \theta,\psi-\imath \mathrm{atanh}(\cos \theta)) \left[ \left(\mathrm{d} t + \frac{r^2+a^2}{\Delta} \mathrm{d} r \right)\wedge \left(\mathrm{d} \varphi+\frac{a}{\Delta}\mathrm{d} r \right) \right. \nonumber \\ & & \left. +\frac{\imath}{\sin \theta} \left(\mathrm{d} t+\frac{\rho^2}{\Delta} \mathrm{d} r -a \sin^2 \theta  \mathrm{d} \varphi \right) \wedge \mathrm{d} \theta \right]
\end{eqnarray}
where $\mathrm{d} u=\mathrm{d} t+\frac{r^2+a^2}{\Delta} \mathrm{d} r$, $\mathrm{d} \psi=\mathrm{d} \varphi+ \frac{a}{\Delta} \mathrm{d} r$ and $f(\cdot, \cdot)$ is an arbitrary, differentiable function of two variables. Only one of Newman--Penrose electromagnetic scalars (\ref{NP}) constructed from $\mathcal{F}$ is non-zero
\begin{equation}
 \phi_{-1}(\mathcal{F})=\frac{2 f(u- \imath a \cos \theta,\psi-\imath \mathrm{atanh}(\cos \theta))}{ \imath \sqrt{\Delta} \sin \theta} \qquad \phi_{0}(\mathcal{F})=\phi_{1}(\mathcal{F})=0 \label{scalar_F}
\end{equation}
Starting from Maxwell equations for Newman--Penrose electromagnetic scalars, it is easy to show that $\phi_{-1}(\mathcal{F})$ given by Eq. (\ref{scalar_F}) fulfils the following equation:
\begin{equation}\label{D-1}
\left(\mathcal{D}^{\dagger}-\frac{r-m}{\sqrt{2 \rho^2 \Delta}}\right)\phi_{-1}(\mathcal{F})=0
\end{equation}
where
\begin{equation}
\mathcal{D}^{\dagger}=\frac{r^2+a^2}{\sqrt{2 \rho^2 \Delta}} \partial_{t}-\sqrt{\frac{\Delta}{2 \rho^2}} \partial_{r}+\frac{a}{\sqrt{2 \rho^2} \Delta} \partial_{\varphi}
\end{equation}
We have also found another family of solutions which satisfies conditions (\ref{self_dual})--(\ref{alg_special}). We denote it by $\mathcal{H}$
\begin{equation}
\mathcal{H}=\frac{2 \imath h(v+\imath a \cos \theta,\phi +\imath \mathrm{atanh}(\cos \theta))}{\sqrt{\Delta} \sin \theta} l \wedge m
\end{equation}
In terms of coordinate forms
\begin{eqnarray}
\mathcal{H}&=&h(v+\imath a \cos \theta,\phi +\imath \mathrm{atanh}(\cos \theta))\left[ \left(\mathrm{d} t - \frac{r^2+a^2}{\Delta} \mathrm{d} r \right)\wedge \left(\mathrm{d} \varphi-\frac{a}{\Delta}\mathrm{d} r \right) \right. \nonumber \\ & & \left. -\frac{\imath}{\sin \theta} \left(\mathrm{d} t+\frac{\rho^2}{\Delta} \mathrm{d} r -a \sin^2 \theta  \mathrm{d} \varphi \right) \wedge \mathrm{d} \theta \right]
\end{eqnarray}
where $\mathrm{d}v=\mathrm{d}t- \frac{r^2+a^2}{\Delta} \mathrm{d}r$, $\mathrm{d} \phi=\mathrm{d} \varphi- \frac{a}{\Delta}\mathrm{d} r$ and $h(\cdot, \cdot)$ is an arbitrary, differentiable function of two variables.
Again, only one Newman--Penrose electromagnetic scalar constructed from $\mathcal{H}$ remains non-vanishing
\begin{equation}
\phi_{-1}(\mathcal{H})=\phi_{0}(\mathcal{H})=0 \qquad \phi_{1}(\mathcal{H})=\frac{h(v+ \imath a \cos \theta,\phi+\imath \mathrm{atanh}(\cos \theta))}{\sqrt{\Delta} \sin \theta} \label{scalar_H}
\end{equation}
note that it is a different component than that for solution 
 (\ref{scalar_F}).
Maxwell equations in terms of Newman--Penrose electromagnetic scalars lead to the following equation
\begin{equation}\label{D1}
\left(\mathcal{D}+\frac{r-m}{\sqrt{2 \rho^2 \Delta}}\right)\phi_{1}(\mathcal{H})=0
\end{equation}
where
\begin{equation}
\mathcal{D}=\frac{r^2+a^2}{\sqrt{2 \rho^2 \Delta}} \partial_{t}+\sqrt{\frac{\Delta}{2 \rho^2}} \partial_{r}+\frac{a}{\sqrt{2 \rho^2} \Delta} \partial_{\varphi}
\end{equation}
{
The formulae (\ref{scalar_F}) and (\ref{scalar_H}) describe explicit examples of Maxwell field, where scalar $\Phi$  (\ref{FI}) vanishes in the whole spacetime.
It means that there exist singular solutions which belong to the kernel of the mapping $F \mapsto \Phi$ (\ref{defPhi}) restricted to the Maxwell solutions.
However, we are convinced that for regular solutions this kernel becomes trivial.}

\subsection{Reconstruction of the Maxwell field from F--I initial data set $(\Phi, \partial_{t} \Phi)$}
In this section we will discuss reconstruction of Maxwell field from given solution of F--I equation (\ref{FI}). It is convenient to use tensor density of electromagnetic field instead of Maxwell tensor $F$.

 Let us introduce the following convention: By small latin letters $(k,l,m...)$ we will denote a three-dimensional space index which corresponds to $(r,\theta, \varphi)$ coordinates. Capital latin letters $(A,B,C...)$ are two-dimensional angular indices which run $(\theta, \varphi)$ subset.

We define respectively electric field density $\mathcal{E}^{k}$, magnetic field density $\mathcal{B}^{k}$ and complex electromagnetic vector field density $\mathcal{Z}^{k}$ in the following way:
\begin{eqnarray}
\mathcal{E}^{k}&:=&\sqrt{-\det g_{\mu \nu}} F^{0 k} \nonumber \\
\mathcal{B}^{k}&:=&-\sqrt{-\det g_{\mu \nu}}  \ast F^{0 k} \label{def_EBZ} \\
\mathcal{Z}^{k}&:=&\mathcal{E}^{k}+\imath \mathcal{B}^{k} \nonumber
\end{eqnarray}
 Let us denote by $\Psi$ a scalar density associated with scalar $\Phi$ (\ref{FI}) by the following formula:
\begin{equation}
\Psi:=\frac{\Sigma^{2}}{(r- \imath a \cos  \theta) \rho^{2}}\frac{1}{a^2+r^2}\sqrt{-\det g_{\mu \nu}} \Phi \label{Psi}
\end{equation}
where $\displaystyle \Sigma:= \sqrt{(a^2+r^2)^2-a^2 \Delta \sin^2 \theta} $,  cf. appendix \ref{Quantity}. $\Psi$ defined by (\ref{Psi}) in the terms of $\mathcal{Z}^{k}$ (\ref{def_EBZ}) takes the form
\begin{equation}
\Psi=\mathcal{Z}^{r}+\imath P_{A}\mathcal{Z}^{A} \label{Psi_Z}
\end{equation}
where
\begin{equation}
P_{\theta}:=\frac{ a \Delta \sin \theta}{r^2+a^2} \qquad \label{Qtheta}
P_{\varphi}:=0
\end{equation}
  We assume that there is given a smooth solution (\ref{Psi}) outside of the exterior horizon $(r>r_{+})$. Maxwell equations (\ref{eqMaxwell}) in the terms of $\mathcal{Z}^{k}$ (\ref{def_EBZ}) can be written as
\begin{eqnarray}
	\partial_{k} \mathcal{Z}^{k}&=&0 \label{dif_Z}\\
	\partial_{t} \mathcal{Z}^{k}&=& \partial_{l}(N^{l}\mathcal{Z}^{k}-N^{k}\mathcal{Z}^{l})-\imath \partial_{l}\left(\frac{\tilde{g}^{km}\tilde{g}^{ln}}{N}\varepsilon_{mnp}\mathcal{Z}^{p}\right) \label{Maxwel_Z}
\end{eqnarray}
where $N=\frac{1}{\sqrt{-g^{00}}}$ is a lapse function (\ref{lapse}) and $N^{k}=N^2 g^{0k}$ is a shift vector (\ref{shift}). It is a time+space decomposition. $\tilde{g}^{km}$ denotes an inverse of three-dimensional space metric. $\varepsilon_{mnp}$ is a three-dimensional  Levi-Civita tensor\footnote{Defined by $\varepsilon_{mnp}=\sqrt{\det g_{kl}} \epsilon_{mnp}$ with convention $\epsilon_{r \theta \varphi}=1$, cf.  (\ref{g3}).}.
Let us notice that (\ref{Psi_Z}) enables one to replace radial component $\mathcal{Z}^{r}$ by $Z^{A}$. Moreover, differentiating Eq. (\ref{Psi_Z}) with respect to $t$ and using Eq. (\ref{Maxwel_Z}), we obtain
\begin{equation}
\imath \left(\partial_{t}-N^{\varphi} \partial_{\varphi} \right) \Psi= \partial_{A}\left(\frac{\Sigma}{\rho} \varepsilon^{AB} \mathcal{Z}_{B}\right)-\imath P_{B}\partial_{A}\left(\frac{\Sigma}{\rho} \varepsilon^{AB} \mathcal{Z}_{r}\right)+\imath P_{A}\partial_{r}\left(\frac{\Sigma}{\rho} \varepsilon^{AB} \mathcal{Z}_{B}\right)
\end{equation}
 $\varepsilon^{AB}$ is a two-dimensional Levi-Civita tensor\footnote{Defined by $\varepsilon_{AB}=\sqrt{\det g_{CD}} \epsilon_{AB}$ with convention $\epsilon_{ \theta \varphi}=1$, cf. (\ref{g2}).}.
  Substituting $\mathcal{Z}^{r}$ with using Eq. (\ref{Psi_Z}), we have
 \begin{equation}
 \imath \left(\partial_{t}-K^{\varphi} \partial_{\varphi} \right) \Psi= \partial_{A}\left(\frac{\Sigma}{\rho} \varepsilon^{AB} \mathcal{Z}_{B}\right)- P_{B}\partial_{A}\left(\frac{{\Sigma \rho}}{\Delta} \varepsilon^{AB} P_{C} \mathcal{Z}^{C}\right)+\imath P_{A}\partial_{r}\left(\frac{\Sigma}{\rho} \varepsilon^{AB} \mathcal{Z}_{B}\right) \label{rot_Z}
 \end{equation}
 where
 \begin{eqnarray}
 K^{\varphi}&=&N^{\varphi}+ \frac{\Sigma}{\rho} g_{rr} P_{\theta} \varepsilon^{\theta\varphi}
  = \frac{a \Delta \rho^{2}}{(r^2+a^2) \Sigma^{2}} \nonumber
 \end{eqnarray}
Differentiating Eq. (\ref{Psi_Z}) with respect to $r$ and using Eq. (\ref{dif_Z}) gives the following:
\begin{equation}
\partial_{r} \Psi=\imath \partial_{r}(P_{C} \mathcal{Z}^{C})-\partial_{A}\mathcal{Z}^{A}
\label{div_Z}
\end{equation}
Equations (\ref{rot_Z}) and (\ref{div_Z}) form a system of differential equations for $\mathcal{Z}^{A}$.
It is a system of first order linear PDE's. For analytic data one can check the local existence of solutions using the Cauchy-Kovalevskaya theorem.
Moreover, the exterior domain $V=[r_0,\infty]\times S^2$ enables one to change this system into the following infinite-dimensional system of ODE's:
\[ \frac{\mathrm{d} u}{\mathrm{d} r}=A(r)u+f\]
where $f,u \in l^2$, $l^2$ is a Hilbert space corresponding to spherical harmonics on $S^2$ and $A(r)$ is one-dimensional family of linear operators in $l^2$. The properties of $A(r)$ determine the system. The vector $f$ corresponds to given data $(\partial_{r} \Psi , \left(\partial_{t}-K^{\varphi} \partial_{\varphi} \right) \Psi )$ and $u$ corresponds to the solution $\mathcal{Z}^{A}$.

Spherical part $\mathcal{Z}^A$ plus Eq. (\ref{Psi_Z}) enables one to reconstruct the full Maxwell field $\mathcal{Z}^{k}$.  The details will be analyzed in a separate paper. 
\vspace{0.3 cm}

{\noindent \sc Acknowledgements} We are grateful to Paweł Nurowski for discussions and helpful remarks.
This work was supported in part by Narodowe Centrum Nauki (Poland) under Grant No.
DEC-2011/03/B/ST1/02625.

\appendix
\section{Proof of equation (\ref{dalamb_Q}) \label{shifted_proof}}
Changing the names of indices, we write (\ref{comT}) three times
\begin{eqnarray}
Q_{\lambda \kappa; \nu \mu}-Q_{\lambda \kappa; \mu \nu}&=&Q_{\sigma \kappa} {R^{\sigma}}_{\lambda \nu \mu}+Q_{\lambda \sigma} {R^{\sigma}}_{\kappa \nu \mu} \nonumber \\
Q_{\mu \kappa; \lambda \nu}-Q_{\mu \kappa; \nu \lambda}&=&Q_{\sigma \kappa} {R^{\sigma}}_{\mu \lambda \nu}+Q_{\mu \sigma} {R^{\sigma}}_{\kappa \lambda \nu} \nonumber \\
Q_{\nu \kappa; \mu \lambda}-Q_{\nu \kappa; \lambda \mu}&=&Q_{\sigma \kappa} {R^{\sigma}}_{\nu \mu \lambda}+Q_{\nu \sigma} {R^{\sigma}}_{\kappa \mu \lambda} \nonumber
\end{eqnarray}
We take the first equation, subtract the second one and finally add the third equation. Assuming that  $Q_{\mu \nu}$ is antisymmetric
\begin{eqnarray}
&& \hspace*{-1cm} Q_{\lambda \kappa; \nu \mu}-Q_{\lambda \kappa; \mu \nu} -Q_{\mu \kappa; \lambda \nu}+Q_{\mu \kappa; \nu \lambda}+Q_{\nu \kappa; \mu \lambda}-Q_{\nu \kappa; \lambda \mu} \nonumber \\
\phantom{X}&=&2 Q_{\lambda \kappa; \nu \mu}-(Q_{\lambda \kappa; \mu}+Q_{\mu \kappa; \lambda})_{;\nu}+(Q_{\mu \kappa; \nu}+Q_{\nu \kappa; \mu})_{; \lambda} -(Q_{\nu \kappa; \lambda}+Q_{\lambda \kappa; \nu})_{; \mu} \nonumber \\
&=&Q_{\sigma \lambda} {R^{\sigma}}_{\kappa \mu \nu}+Q_{\sigma \mu} {R^{\sigma}}_{\kappa \lambda \nu}+Q_{\sigma \nu} {R^{\sigma}}_{\kappa \lambda \mu}+ 2 Q_{\sigma \kappa} {R^{\sigma}}_{\mu \nu \lambda} \label{Q_brackets}
\end{eqnarray}
We denote the covector $\xi_{\mu}=\nabla^{\sigma}Q_{\sigma \mu}$. It fulfils
\begin{equation}
2{\xi^{\mu}}_{;\mu}={Q^{\sigma \rho}}_{;\sigma \rho}-{Q^{\sigma \rho}}_{;\rho \sigma}=-2 Q^{\kappa \nu} R_{\kappa \nu}=0 \label{div_xi}
\end{equation}
In the above equality we use Eq. (\ref{comT}). Definition of the CYK two-form (\ref{df_Q}) applied to the terms in brackets in the Eq. (\ref{Q_brackets}) implies
\begin{eqnarray}
2Q_{\lambda \kappa; \nu \mu}&=&\frac{2}{3} (g_{\lambda \mu} \xi_{\kappa;\nu} +g_{\nu \lambda} \xi_{\kappa; \mu}-g_{\mu \nu} \xi_{\kappa; \lambda}-g_{\kappa(\lambda} \xi_{\mu);\nu}+g_{\kappa(\mu} \xi_{\nu);\lambda} -g_{\kappa(\nu} \xi_{\lambda);\mu}) \nonumber \\
&&+Q_{\sigma \lambda} {R^{\sigma}}_{\kappa \mu \nu}+Q_{\sigma \mu} {R^{\sigma}}_{\kappa \lambda \nu}+Q_{\sigma \nu} {R^{\sigma}}_{\kappa \lambda \mu}+ 2 Q_{\sigma \kappa} {R^{\sigma}}_{\mu \nu \lambda} \nonumber \\
&&+\underbrace{\mathcal{Q}_{\lambda \kappa \mu;\nu}}_{0}-\underbrace{\mathcal{Q}_{\mu \kappa \nu;\lambda}}_{0}+\underbrace{\mathcal{Q}_{\nu \kappa \lambda; \mu}}_{0} \label{war_cal}
\end{eqnarray}
Contracting (\ref{war_cal}) with respect to indices $\mu$ and $\nu$ and using the algebraic properties of $\mathcal{Q}$, we get
\begin{equation}
{{Q_{\lambda\kappa}}^{;\mu}}_{\mu}+{R^{\sigma}}_{\kappa \lambda \mu} {Q^{\mu}}_{\sigma}+Q_{\sigma \kappa} {R^{\sigma}}_{\lambda}+\frac{2}{3} \xi_{(\kappa;\lambda)}+\frac{1}{3} g_{\kappa \lambda} {\xi^{\mu}}_{;\mu}= \underbrace{{\mathcal{Q}_{\mu \kappa \lambda}}^{;\mu}}_{0} \label{war_cal_con}
\end{equation}
Using (\ref{comT}) and (\ref{df_Q}) leads to
\begin{equation}
 \xi_{(\mu;\lambda)}=\frac{3}{2}R_{\sigma (\mu} {Q_{\lambda)}}^{\sigma} \label{sym_xi}
\end{equation}
Combining equations (\ref{war_cal_con}), (\ref{sym_xi}) and (\ref{div_xi}), we obtain
\begin{equation}
\Box Q_{\lambda \kappa}={R^{\sigma}}_{\kappa \lambda \nu} {Q_{\sigma}}^{\nu}-R_{\sigma[\kappa}{Q_{\lambda]}}^{\sigma}
\end{equation}
where $\Box T_{\rho \nu}=\nabla_{\mu} \nabla^{\mu} T_{\rho \nu}$.
\section{Useful quantities \label{Quantity}}
\paragraph{Kerr metric} Non-trivial components of the inverse metric are
\begin{eqnarray}
g^{tt}&=&- \frac{\Sigma^2}{\rho^{2} \Delta} \label{gtt}\\
g^{t \varphi}&=&-\frac{2 m r a}{\rho^{2} \Delta} \label{gtf}
\end{eqnarray}
where
\begin{equation*}
\Sigma^{2}=(a^2+r^2)^2-a^2 \Delta \sin^2 \theta
\end{equation*}
The three-metric has the following diagonal components:
\[ g_{rr}= \frac{\rho^2}{\Delta} \quad g_{\theta\theta}= \rho^2 \quad g_{\varphi\varphi}= \frac{\Sigma^2}{\rho^2}\sin^2\theta \]
Lapse function $N$ and the non-vanishing component of the shift vector $N^{k}$ in $3+1$ decomposition of the Kerr metric (\ref{kerr.metric}) are
\begin{eqnarray}
N&=&\sqrt{\frac{\rho^{2} \Delta}{\Sigma^2}} \label{lapse}\\
N^{\varphi}&=&-\frac{2 m r a}{\Sigma^2} \label{shift}
\end{eqnarray}
Square roots of determinants of Kerr metric for the corresponding dimensions are
\begin{equation}
\sqrt{-\det g_{\mu\nu}}= \rho^2 \sin \theta
\end{equation}
\begin{equation}
	\sqrt{\det g_{kl}}=\sqrt{\frac{\rho^2}{\Delta}} \Sigma \sin \theta \label{g3}
\end{equation}
\begin{equation}
	\sqrt{\det g_{AB}}=\Sigma \sin \theta \label{g2}
\end{equation}
 In Kerr spacetime (\ref{kerr.metric}) Maxwell field tensor density is related to electric $\mathcal{E}^i$ and magnetic $\mathcal{B}^{i}$ field densities in the following way
\begin{eqnarray}
\tilde{F}&=& \sqrt{-\det g_{\mu \nu}} F \nonumber \\
 &=& \mathcal{E}^i \partial_{t} \wedge  \partial_{i}+\mathcal{B}^{\phi} \frac{\Delta \sin \theta}{\rho^2} \partial_{r} \wedge \partial_{\theta} \nonumber
 -\left(\mathcal{B}^{\theta} \frac{\Delta \rho}{\Sigma^2 \sin \theta }+\mathcal{E}^{r} \frac{2 a m r }{\Sigma^2}  \right) \partial_{r} \wedge \partial_{\phi}  \\
& & +\left(\mathcal{B}^{r} \frac{\Delta \rho}{\Sigma^2 \sin \theta }-\mathcal{E}^{\theta} \frac{2 a m r }{\Sigma^2}  \right) \partial_{\theta} \wedge \partial_{\phi}
\end{eqnarray}


\end{document}